\documentclass[a4paper, 11pt]{article}

\usepackage{amsmath,amsfonts,amssymb,amsthm,graphicx,graphics,epsf}
\usepackage{graphicx}

\usepackage{algorithmic}
\usepackage{algorithm}
\usepackage{hyperref,color}
\usepackage{xspace}

\usepackage{geometry}

\newtheorem{theorem}{Theorem}
\newtheorem{lemma}[theorem]{Lemma}
\newtheorem{proposition}[theorem]{Proposition}
\newtheorem{corollary}[theorem]{Corollary}
\theoremstyle{definition}

\newtheorem{observation}[theorem]{Observation}

\theoremstyle{remark}

\begin{document}

\title{Redundancies in Linear Systems with two Variables per Inequality }

\author{Komei Fukuda \thanks{Department of Mathematics and Department of Computer Science. Institute of Theoretical Computer Science,
 ETH Z\"{u}rich.  CH-8092 Z\"{u}rich, Switzerland.
  \texttt{komei.fukuda@math.ethz.ch}}
  \and 
  May Szedl{\'a}k\thanks{Department of Computer Science. Institute of Theoretical Computer Science, ETH Z\"{u}rich.
  CH-8092 Z\"{u}rich, Switzerland.
    \texttt{may.szedlak@inf.ethz.ch} \newline Research supported by the Swiss National Science Foundation (SNF Project 200021\_150055 / 1)}
  }

\date{October 10, 2016}

\maketitle

\begin{abstract}
The problem of detecting and removing redundant constraints is fundamental in optimization. We focus on the case of linear programs (LPs), given by $d$ variables with $n$ inequality constraints. A constraint is called \emph{redundant}, if after its removal, the LP still has the same feasible region. The currently fastest method to detect all redundancies is due to Clarkson: it solves  $n$ linear programs, but each of them has at most $s$ constraints, where $s$ is the number of nonredundant constraints.

In this paper, we study the special case where every constraint has at most two variables with nonzero coefficients. This family, denoted by $LI(2)$, has some nice properties. Namely, as shown by Aspvall and Shiloach, given a variable $x_i$ and a value $\lambda$,  we can test in time $O(nd)$ whether there is a feasible solution with $x_i = \lambda$.
Hochbaum and Naor present an $O(d^2 n \log n)$ algorithm for solving the feasibility problem in $LI(2)$. Their technique makes use of the Fourier-Motzkin elimination method and the earlier mentioned result by Aspvall and Shiloach. 

We present a strongly polynomial algorithm that solves redundancy detection in time $O(n d^2 s \log s)$. It uses a modification of Clarkson's algorithm, together with a revised version of Hochbaum and Naor's technique.
Finally we show that dimensionality testing can be done with the same running time as solving feasibility.
\end{abstract}

\section{Introduction}
The problem of detecting and removing redundant constraints is fundamental in optimization. Being able to understand redundancies in a model is an important step towards improvements of the model and faster solutions.

Throughout we consider linear systems of inequalities of form $Ax \leq b$, for $A \in R^{n \times d}$, $b \in R^n$, $d<n$. 
The $j$-th constraint, denoted $A_jx \leq b_j$,  
is called \emph{redundant}, if its removal does not change the set of feasible solutions. By removing $A_jx \leq b_j$ from the system  we get a new system denoted $A_{[n] \setminus \{j\}}x \leq b_{[n] \setminus \{j\}}$. Assume that $Ax \leq b$ is feasible, then by solving the following linear program (LP) we can decide redundancy of  $A_jx \leq b_j$.
\begin{equation}\label{eq:LP_r}
\begin{array}{lrcl}
\mbox{maximize} & A_jx  \\ 
\mbox{subject to} &  A_{[n] \setminus \{j\}}x &\leq& b_{[n] \setminus \{j\}}. 
\end{array}
\end{equation}
Namely, a constraint $A_jx \leq b_j$ is redundant if and only if the optimal solution has value at most $b_j$.

Let $LP(n,d)$ denote the time needed to solve an LP with $n$ inequalities and $d$ variables. Solving $n$ linear programs of form (\ref{eq:LP_r}), with running time $LP(n,d)$, is enough for detecting all redundancies. The currently fastest method is due to Clarkson with running time $\mathcal{O}(n\cdot LP(s,d))$~\cite{c-mosga-94}, where we initially assume an interior point is given. This method also solves $n$ linear programs, but each of them has at most $s$ variables, where $s$ is the number of nonredundant variables. Hence, if $s\ll n$, this \emph{output-sensitive} algorithm is a major improvement.

In general no strongly polynomial time algorithm (polynomial in $d$ and $n$) to solve an LP is known. Although the simplex algorithm runs fast in practice, in general it can have exponential running time \cite{dantzig-simplex, klee-minty-simplex}. On the other hand the ellipsoid method runs in polynomial time on the encoding of the input size, but is not practical \cite{Khachiyan}. A first practical polynomial time algorithm, the interior-point method, was introcuded in \cite{Karmarkar}, and has been modified in many ways since.

In this paper we focus on the special case where every constraint has at most two variables with nonzero coefficients, we denote this family by $LI(2)$. 
Our main result is that for a full-dimensional system $Ax \leq b$ in $LI(2)$ we can detect all redundancies in time $O(nd^2 s \log s )$ (see Theorem \ref{thm_main}), where we assume that an interior point solution is given. To our knowledge, this is a first strongly polynomial time algorithm for redundancy detection in $LI(2)$. 

To obtain this running time we use an alternated version of Clarkon's algorithm, which solves feasibility problems instead of optimization problems.
Moreover our algorithm makes use of a modified version of Hochbaum and Naor's algorithm, which for a system in $LI(2)$   finds a feasible point or a certificate for infeasibility in time $O(d^2n \log n)$ \cite{HN}. This result is an improvement of Megiddo's algorithm with running time $O(d^3 n \log n)$ \cite{Megiddo}. Although their techniques are similar and both rely heavily on \cite{Aspvall}, the improved version is much simpler.

 We will give a summary of the Hochbaum-Naor Algorithm in Section \ref{ch_HN}. 
In Section \ref{sec_HNmod} we will give a stronger version of this algorithm, which decides full-dimensionality and in the full-dimensional case outputs an interior point. Using this variant of the algorithm together with our modification of Clarkson's algorithm we get an output sensitive, strongly polynomial time redundancy detection algorithm. 
In Section  \ref{sec_nonfull} we show how the results extend to non-full-dimensional systems (see Theorem \ref{thm_nonred}).  

In all cases the preprocessing can also be done in strongly polynomial time. Moreover, we show that dimensionality testing of a polytope $P = \{x \mid Ax \leq b\}$ can be done with the same running time as the feasibility testing method of Hochbaum and Naor (see Corollary \ref{cor_dim}). Note that for general LP's one needs to solve up to $d$ optimization problems. 

Although in $LI(2)$ one can find a feasible solution fast, it is not known how to find an optimal solution in strongly polynomial time. For general LPs a standard technique for converting an optimization problem into a feasibility problem is to use the dual linear program. However, the dual of a system in $LI(2)$ is generally not in $LI(2)$. If the objective function is in $LI(2)$, one can apply binary search on the value of the optimal solution, this gives an algorithm that depends on the input size. 

Note that Clarkson's algorithm relies on finding an optimal solution of a linear program. Since for $LI(2)$ we do not have a fast way to optimize, this is the reason why we modify the algorithm such that it only solves feasibility problems. 

\section{Definitions and Preliminaries} 
As already mentioned in the introduction we throughout consider \emph{linear systems} of the form 
\[Ax \leq b,\]
where $A \in R^{n \times d}$, $b \in R^n$.

The set of inequalities of $Ax \leq b$ is denoted by $G$. 
A point $x^* \in R^d$ is a \emph{feasible solution} or \emph{feasible point} of $Ax \leq b$ (or $G$) if $Ax^* \leq b$. It is called an \emph{interior point solution} of $Ax \leq b$ (or $G$) if all inequalities are satisfied with strict inequality, i.e., $Ax^* < b$ (where "$<$" denotes the componentwise strict inequality). The system $Ax \leq b$ (or $G$) is called \emph{feasible} if a feasible solution exists, otherwise it is called \emph{infeasible}. If an interior point solution exists, the system is called \emph{full-dimensional}. The system $Ax \leq b$ is called $k$-dimensional if the solution set $\{x \mid Ax \leq b\}$ is $k$-dimensional.

For a subset $ S \subseteq [n]:=\{1,\dots, n\}$ we denote by $A_Sx \leq b_S$ the subsystem of $Ax \leq b$ containing only the inequalities indexed by $S$. In particular the $j$-th constraint is denoted by $A_jx \leq b_j$. This constraint is called \emph{redundant} if  $A_{[n] \setminus \{j\}}x \leq b_{[n] \setminus \{j\}}$ implies $A_jx \leq b_j$ or equivalently if there is no solution to the system 
\begin{equation*}\label{eq:system}
\begin{array}{rcl}
   A_{[n] \setminus \{j\}}x &\leq& b_{[n] \setminus \{j\}} \\
   A_{j}x &>& b_{j}.
\end{array}
\end{equation*}
As mentioned in the introduction, we can test redundancy of a constraint $A_jx \leq b_j$ by solving an LP of form (\ref{eq:LP_r}).

For a feasible system $Ax \leq b$ and a variable $x_i$ let $[x_i^{\min}, x_i^{\max}]$ be the projection of the solution space of $Ax \leq b$ to the $x_i$-axis, we call this the \emph{range} of $x_i$. This interval is exactly the set of values of $x_i$ for which a solution of the entire system can be constructed. It is possible that $x_i^{\min} = -\infty$ or $x_i^{\max} = \infty$.

In this paper we are interested in sparse linear systems, in particular the family $LI(2)$. A linear system is in $LI(2)$, if every constraint has at most two variables with nonzero coefficients. That means all inequalities have form $\alpha x_i + \beta x_j \leq \gamma$ for some $\alpha$, $\beta$, $\gamma \in R$, $\alpha \neq 0$.

We define the \emph{neighbors} of $x_i$ in $G$, denoted $N(x_i, G)$, as the set of variables $x_j$, $j \in [d] \setminus \{i\}$ for which there exists an inequality in $G$ containing $x_i$ and $x_j$ with nonzero coefficients. 

The system $(Ax \leq b)|_{x_i =c}$ (or $G |_{x_i =c}$) is obtained form $Ax \leq b$ (or $G$) by substituting the variable $x_i$ by the constant $c$, it hence has one variable less than the original system.




\section{A Strongly Polynomial Time  Redundancy Detection Algorithm for Linear Programs with two Variables per Inequality}
In this section we will prove our main result, the running time of the strongly polynomial algorithm to detect all redundancies in $LI(2)$ (see Theorem \ref{thm_main}). 

\begin{figure*}[h]
\begin{center}
\includegraphics{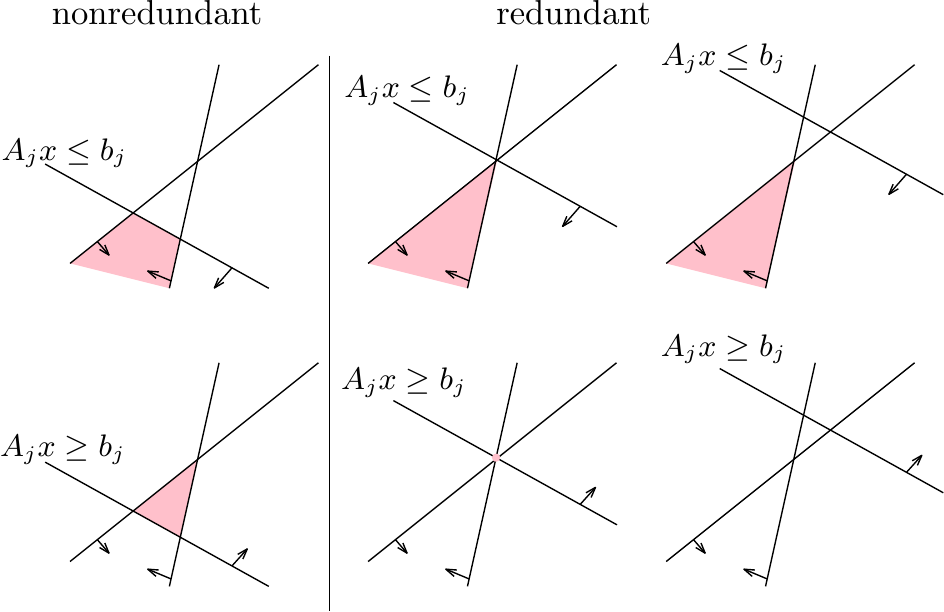}
\end{center}
\caption{Redundancy Certificates} \label{fig_redcases}
\end{figure*}

We make use of the following modified version of Hochbaum and Naor's result (Theorem \ref{thm_modHN}). We will discuss their original result in Section \ref{sec_HN} and the validity of the modification in Section \ref{sec_HNmod}. In Section \ref{sec_nonfull} we will discuss how the results extend to non-full-dimensional systems.

\begin{theorem} 
\label{thm_modHN}
For a system $Ax \leq b$ in $LI(2)$ one can decide in time $O(d^2 n \log n)$ whether the system is full-dimensional, and in the full-dimensional case output an interior point solution.
\end{theorem}

\begin{theorem}
\label{thm_main}
Let $Ax \leq b$ a full-dimensional system in $LI(2)$. Let $z$ be an interior point solution of $Ax \leq b$ and let $z^\epsilon := z+(\epsilon, \dots, \epsilon^d)^T$ for some $\epsilon$ small enough, i.e., $z^\epsilon$ is a generic interior point. Then the following algorithm detects all redundancies in time $O(nd^2s\log s)$.

 \medskip
\noindent
\begin{tabbing}
123 \= 123 \= 123 \= 123 \= 123 \= 123 \= 123 \= 123 \= \kill
\> Algorithm Modified Clarkson ($A$,$b$,$z$);\\
\> {\bf begin }\\
\> \> $R := \emptyset, S := \emptyset$;\\
\> \> {\bf while} $R \cup S \neq [n]$ {\bf do}\\
\> \> \> pick any $r \in [n] \setminus (R \cup S)$ and use Theorem \ref{thm_modHN} on $\{A_Sx \leq b_S$, $a_r x\geq b_r\}$; \\
\> \> \> {\bf if} $A_Sx \leq b_S$, $a_rx \geq b_r$ not full-dimensional {\bf then} \\
\> \> \> \> $R = R \cup \{r\}$;\\
\> \> \> {\bf else}  let $x^*$ be an interior point solution  of $A_Sx \leq b_S$, $a_r x\geq b_r$ {\bf then} \\
\> \> \> \> $S = S \cup \{q\}$, where $q = \text{RayShoot}(A,b,z^\epsilon,x^*-z^\epsilon)$; \\
\> \> \> {\bf endif};\\
\> \> {\bf endwhile}; \\
\> \>  output $S$; \\ 
\> {\bf end}.\\       
\end{tabbing}
The function RayShoot$(A,b,z,t)$ returns the index $q$ of a facet-inducing hyperplane $\{x: A_qx = b_q\}$, which is hit first by the ray starting at $z$ along the direction of $t$. 
\end{theorem}

Note that Theorem \ref{thm_modHN} immediately implies that the interior point $z$ of Theorem \ref{thm_main} can be found in strongly polynomial time  $O(nd^2\log n)$. It follows that finding all redundancies and the preprocessing can be achieved in strongly polynomial time.

Using Theorem \ref{thm_modHN} and the following observation (see also Figure \ref{fig_redcases}) we can prove Theorem \ref{thm_main}.

\begin{observation}
\label{obs_fulldim}
Let $Ax \leq b$ be a full-dimensional system in $LI(2)$. Then for $j \in [n]$ the following are equivalent.
\begin{enumerate}
	\item $A_j x \leq b_j$ is nonredundant in $Ax \leq b$.
	\item $A_j x \leq b_j$ is facet-inducing for $P = \{x \mid Ax \leq b\}$.
	\item The system  $A_{[n]\setminus \{j\}}x \leq b_{[n]\setminus \{j\}}$, $A_j x \geq b_j$ is full-dimensional.
\end{enumerate}
\end{observation}

\begin{proof} [Proof of Theorem \ref{thm_main}.]
We have to show that the modified Clarkson Algorithm returns $S'$, the indices of the set of nonredundant constraints. 
We first discuss correctness of the algorithm by induction. We claim that in every step $S \subseteq S'$ and $R \subseteq [n] \setminus S'$. This is trivially true in the beginning. Assume that in some step of the algorithm we have $S \subseteq S'$, $R \subseteq [n] \setminus S'$ and $r \in [n] \setminus (S \cup R)$. If $A_Sx \leq b_S$, $a_rx \geq b_r$ is not full-dimensional, then $A_{S'}x \leq b_{S'}$, $A_r x \geq b_r$ is not full-dimensional and hence $r$ is redundant by Observation \ref{obs_fulldim}. 

If $A_Sx \leq b_S$, $a_r x\geq b_r$ is full-dimensional and $x^*$ is an interior point, then we do ray shooting from $z^\epsilon$ to $x^*$. Note that $a_rx^* > b_r$, hence $x^*$ is not in the feasible region of $Ax \leq b$. Then the first constraint hit (with index $q$) is nonredundant. This constraint is unique, since $z^\epsilon$ is generic and $x^*$ is not a feasible solution. Denote the intersection point of the hyperplane given by $A_qx = b_q$ with the ray by $y^*$.  It follows that $q \notin S$ since $A_qy^* = b_q$ and we know that $A_Sy^* < b_S$. 
This proves correctness of the algorithm.

It remains to discuss the running time. 
Since in every round we add either a variable to $S$ or $R$, the outer loop is executed $n$ times. In every round we run the Algorithm of Theorem \ref{thm_modHN} on at most $s$ inequalities. This takes time $O(nd^2 s \log s)$ by Theorem \ref{thm_modHN}. Moreover there are at most $s$ stages of ray shooting which takes $O(nds)$ time in total. The running time follows. 
\end{proof}


\section{Revision of the Hochbaum-Naor Method}
\label{ch_HN}

Since we modify Hochbaum and Naor's Method in the next section, for completeness we review the basic components and the key ideas of the algorithm. 
\label{sec_HN}
\begin{theorem}
\label{thm_HN} \cite{HN}
For a system $Ax \leq b$ in $LI(2)$ one can decide  whether the problem is feasible in time $O(d^2 n \log n)$, and in the feasible case output a solution.
\end{theorem}

In Section \ref{sec_ingredients} we will give all relevant tools to prove the theorem. We then discuss the feasible case in Sections \ref{sec_feasible} and \ref{sec_feasibleproof}, and finally the infeasible case in Section \ref{sec_infeasible}

\subsection{The Ingredients}
\label{sec_ingredients}
The Hochbaum-Naor Theorem is a mix of an efficient implementation of the Fourier-Motzkin method and the result by Aspvall and Shiloach described below \cite{Aspvall}. In general the Fourier-Motzkin method may generate an exponential number of inequalities, however in the $LI(2)$ case one can implement it efficiently.

\paragraph{The Fourier-Motzkin Method (for general LPs)} (For more details please refer to \cite[pp. 155 - 156]{Schrijver}). Let $G$ be a set of inequalities on the variables $x_1, \dots, x_d$. The Fourier-Motzkin Method eliminates the variables one by one to obtain a feasible solution or a certificate of infeasibility. At step $i$, the LP only contains variables $x_i, \dots, x_n$. Let us denote this system by $G_i$. To eliminate variable $x_i$ all inequalities that contain $x_i$ are written as $x_i \leq h$ or $x_i \geq \ell$, where $h$ and $\ell$ are some linear functions in $x_{i+1}, \dots, x_d$.  Let us denote the two families of inequalities obtained by $H$ and $L$, respectively. For each $h \in H$ and each $\ell \in L$ we add a new inequality $\ell \leq h$. This yields a set of inequalities $G_{i+1}$, on the variables $x_{i+1}, \dots , x_n$. The method is feasibility preserving and given a solution to $G_{i+1}$, one can construct a solution of $G_i$ in time $O(|G_i|)$.


\paragraph{The Aspvall-Shiloach Method \cite{Aspvall}}
Hochbaum and Naor's algorithm highly relies on a the following result by Aspvall and Shiloach \cite{Aspvall}. 

For $Ax \leq b \in LI(2)$ let $g$, $h \in G$ be of form $g: \alpha_1 x + \beta_1y \leq \gamma_1$, $h: \alpha_2 y + \beta_2z \leq \gamma_2$, for $\alpha_1, \alpha_2, \beta_1 \neq 0$, i.e., $g$ and $h$ share a variable. If $\beta_1 > 0$ and $\alpha_2 <0$, one can \emph{update $g$ with $h$} to get a bound on $x$ in terms of $z$ as follows:

Assume that $\alpha_1 >0$, it follows that $x \leq \frac{\gamma_1 - \beta_1y}{\alpha_1}$ and $y \geq \frac{\gamma_2 - \beta_2 z}{\alpha_2}$, and hence
\[x \leq \frac{\gamma_1 - \beta_1y}{\alpha_1} \leq \frac{\gamma_1 - \beta_1\frac{\gamma_2 - \beta_2 z}{\alpha_2}}{\alpha_1}.\]

In the case that $\alpha_1 < 0$ we get a lower bound on $x$ in terms of $z$. 
If  $\beta_1 > 0$ and $\alpha_2 <0$ one can similarly update $x$. If the the signs of $\beta_1$ and $\alpha_2$ are the same, then no update on $x$ is possible. 

For a family of $\{g_i: \alpha_i x_i + \beta_i y_i \leq \gamma_ i\} \subseteq G$, $i = 1, \dots, k$ ($k \geq 1$) we can do a sequence of updates to the bound on $x_i$, iff for all $i = 1, \dots, k-1$, $y_i = x_{i+1}$ and $\text{sign}(\beta_i) = -\text{sign}(\alpha_{i+1}) \neq 0$. If $\beta_k = 0$ this is called a \emph{chain} of length $k$. If $\beta_k \neq 0$ and $y_k = x_1$, this is called \emph{cycle} of length $k$. A chain or a cycle is called \emph{simple}, if every inequality appears at most once. (This concept was first introduced in \cite{Shostak}).

For example $g_1: x - \frac{1}{3}y \leq 1$, $g_2: y - \frac{1}{4}z \leq 1$, $g_3: z \leq 1$ defines a chain of length $3$ as follows:
\[x \leq 1 + \frac{1}{3}y \leq 1 + \frac{1}{3}(1 + \frac{1}{4}z) \leq 1 + \frac{1}{3}(1 + \frac{1}{12}) = \frac{17}{12}.\]

Similarly $g_1: x - \frac{1}{3}y \geq \frac{1}{12}$, $g_2: y - \frac{1}{4}z \geq 0$, $g_3: z - \frac{1}{3}x \geq 0$ defines a cycle of length $3$, 
\[x \geq \frac{1}{12} + \frac{1}{3}y \geq \frac{1}{12} + \frac{1}{3}\cdot\frac{1}{4}z \geq \frac{1}{12} + \frac{1}{3} \cdot \frac{1}{4} \cdot \frac{1}{3} x, \]
which implies $x \geq \frac{3}{35}$.
A cycle results in an inequality of form $x \leq c + ax$ (or $x \geq c + ax$). If  $a=1$, $c <0$ (or $a =1$, $c>0$, respectively), then this is an infeasibility certificate for $Ax \leq b$.

For a variable $x_i$ we denote by $x_i^\text{low}$ ($x_i^\text{high}$), the best lower (upper) bound that can be obtained by considering all chains and cycles of length at most $d$. If the system $Ax \leq b$ is feasible, one can show that $[x_i^\text{low}, x_i^\text{high}] = [x_i^\text{min}, x_i^\text{max}] $. 
Recall that the interval $[x_i^\text{min}, x_i^\text{max}] $ denotes the range of $x_i$, $x_i^\text{min}$ ($x_i^\text{max}$) is the smallest (largest) value that $x_i$ can take such that the system still has a feasible solution.
The direction $[x_i^\text{min}, x_i^\text{max}]  \subseteq [x_i^\text{low}, x_i^\text{high}] $ follows immediately from the definitions. For the other direction let $p$ be a point of $P = \{x \mid Ax \leq b\}$ such that $p_i = x_i^\text{min}$. Looking at the simple chains and simple cycles of the halfspaces that contain $p$ on their boundary, gives us the lower bound $x_i^\text{min}$. The case for $x_i^\text{high}$ and $x_i^\text{max}$ is equivalent.

 On the other hand if the system is infeasible then two things can happen. If $x_i^\text{high} < x_i^\text{low}$, then the range of $x_i$ is empty and this is a certificate for infeasibility. Such a certificate may not exist in general. This is the case for example if the linear system consists of two independent subsystems, one feasible and one infeasible. Also for the small (infeasible) example 
$y+z \leq -1$, 
$y+z \geq 1$, 
$x+y \leq 1 $,
we have that
$x^\text{low} = -\infty$ and $x^\text{high} = \infty$.

From now on we will only discuss the case where $Ax \leq b$ is feasible. In Section \ref{sec_infeasible} we will discuss how to get an infeasibility certificate using the same algorithm. 
\begin{theorem} \cite{Aspvall} \label{proc_AS} 
Given a feasible system $Ax \leq b$ in $LI(2)$, a variable $x_i$ and a value $\lambda \in R$, one can decide in time $O(nd)$ whether
	$\lambda < x_i^\text{min}$,
	$\lambda = x_i^\text{min}$,
	$\lambda \in (x_i^\text{min}, x_i^\text{max})$,
	$\lambda = x_i^\text{max}$ or
	$\lambda > x_i^\text{max}$.
\end{theorem}

Therefore in the feasible case this algorithm decides whether $\lambda$ lies in the open  range $(x_i^\text{min}, x_i^\text{max})$, on boundary of the range, or outside of the range (and on which side). 

The proof of Theorem \ref{proc_AS} requires many technical details. In the following we will summarize the method and provide an intuitive idea. For detailed proofs refer to \cite{Aspvall, HN}.

Let $x_i$ be a fixed variable and $\lambda \in R$. For all $j \in [n]$ let us denote by lo$(x_j)$ (up$(x_j)$) the trivial lower (upper) bound on $x_j$ given by $G$, which may also be infinite. 

The following algorithm fixes $x_i=\lambda$ and returns upper and lower bounds on $x_i$ accordingly. In $d$ rounds, it updates the lower and upper bounds, denoted by $\underline{x}_j$ and $\overline{x}_j$, respectively, on all variables, where initially we are given $\underline{x}_j := \text{lo}(x_j)$ and $\overline{x}_j := \text{up}(x_j)$.
\medskip
\noindent
\begin{tabbing}
123 \= 123 \= 123 \= 123 \= 123 \= 123 \= 123 \= 123 \= \kill
\> Algorithm Aspvall-Shiloach ($G$, $i$, $\lambda$);\\
\> {\bf begin }\\
\> \> \textbf{for} $j=1, \dots, d$ \textbf{do}\\
\> \> \>  $\underline{x}_j := \text{lo}(x_j), \overline{x}_j := \text{up}(x_j)$;\\
\> \> \textbf{endfor};\\
\> \> $\underline{x}_i := \max\{\underline{x}_i,\lambda\}$, $\overline{x}_i := \min\{\overline{x}_i,\lambda\}$ ;\\
\> \> {\bf for} $j = 1, \dots, d$ {\bf do}\\
\> \> \> {\bf for} $g \in G$, with $g: ax_j + bx_k \leq c$, $a,b \neq 0$ {\bf do} \\
\> \> \> \> \textbf{if} $a >0, b>0$ \textbf{then} \\
\> \> \> \> \> $x_j \leq \frac{c - b\underline{x}_k}{a}=: \overline{x}_j^g$, $x_k \leq \frac{c - a\underline{x}_j}{b}=: \overline{x}_k^g$ ;\\
\> \> \> \> \textbf{elseif} $a >0, b<0$ \textbf{then} \\
\> \> \> \> \> $x_j \leq \frac{c - b\overline{x}_k}{a}=: \overline{x}_j^g$, $x_k \geq \frac{c - a\overline{x}_j}{b}=: \underline{x}_k^g$; \\
\> \> \> \> \textbf{elseif} $a <0, b>0$ \textbf{then} \\
\> \> \> \> \> $x_j \geq \frac{c - b\overline{x}_k}{a}=: \underline{x}_j^g$, $x_k \leq \frac{c - a\overline{x}_j}{b}=: \overline{x}_k^g$ ;\\
\> \> \> \> \textbf{else} /*$a <0, b<0$*/ \textbf{then} \\
\> \> \> \> \> $x_j \geq \frac{c - b\underline{x}_k}{a}=: \underline{x}_j^g$, $x_k \geq \frac{c - a\underline{x}_j}{b}=: \underline{x}_k^g$ ;\\
\> \> \> {\bf end for};\\
\> \> \> {\bf for} $\ell = 1, \dots, d$ {\bf do}\\
\> \> \> \> $\underline{x}_\ell:= \max_g\{\underline{x}_\ell, \underline{x}_\ell^g\}$, $\overline{x}_\ell:= \min_g\{\overline{x}_\ell,  \overline{x}_\ell^g\}$;\\
\> \> \> {\bf endfor};\\
\> \> {\bf endfor};\\
\> \> output $\underline{x}_i, \overline{x}_i$; \\
\> {\bf end}.\\       
\end{tabbing}
The algorithm runs in $d$ rounds of $O(n)$ steps each, which results in the running time of $O(nd)$.

\begin{figure*}[h]
\begin{center}
\includegraphics{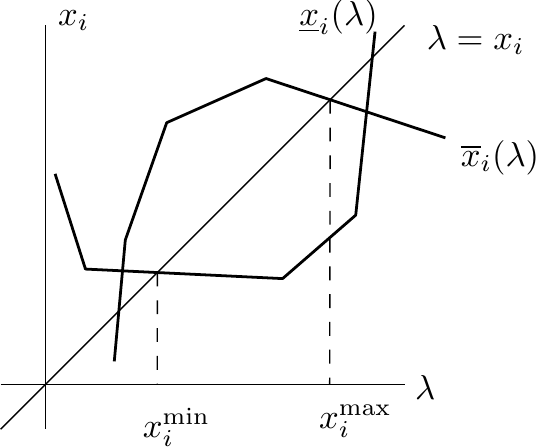}
\end{center}
\caption{Output Aspvall-Shiloach} \label{fig_AS}
\end{figure*}

Since the output of the algorithm depends on $\lambda$ we denote the outputs by $\underline{x}_i(\lambda)$ and $\overline{x}_i(\lambda)$. By the properties of the algorithm it follows that the function $\underline{x}_i(\lambda)$ is convex, whereas $\overline{x}_i(\lambda)$ is concave (see Figure \ref{fig_AS}).

One can show that $\lambda \in [x_i^\text{min}, x_i^\text{max}]$, if and only if $\underline{x}_i(\lambda) \leq \lambda \leq \overline{x}_i(\lambda)$. It follows that $x_i^\text{min} =\min\{\lambda \mid \underline{x}_i(\lambda) \leq \lambda \leq \overline{x}_i(\lambda)\}$ and $x_i^\text{max} =\max\{\lambda \mid \underline{x}_i(\lambda) \leq \lambda \leq \overline{x}_i(\lambda)\}$.  To distinguish between all cases of Theorem \ref{proc_AS}, we additionally need the (left and right) slopes of $\underline{x}_i(\lambda)$ and $\overline{x}_i(\lambda)$. It is not hard to modify the Aspvall-Shiloach Algorithm in such a way that it keeps track of the slopes as well. For instance if $ \underline{x}_i(\lambda) = \lambda \leq \overline{x}_i(\lambda)$ and the slope of $\underline{x}_i$ is smaller than one at $\lambda$, then $\lambda=x_i^\text{min}$. In the case where the slope is greater than one  $\lambda=x_i^\text{max}$.  
Using a careful case distinction one can show that the for given $\lambda$, the values of $\underline{x}_i(\lambda)$ and $\overline{x}_i(\lambda)$ and their (left and right) slopes at $\lambda$ are enough to decide all cases of Theorem \ref{proc_AS}.

\subsection{The Algorithm for the Feasible Case}
\label{sec_feasible}
The rough idea of the algorithm is the following. At step $i$ we want to efficiently find $\lambda$ in the current range $[x_i^{\min}, x_i^{\max}]$ and set $x_i = \lambda$ to obtain a system with one less variable. Whenever this is not possible, we eliminate $x_i$ efficiently in a Fourier-Motzkin step. After this first part we set all variables that were eliminated to values in their current range in the normal Fourier-Motzkin backtracking step.

\paragraph{First Part} The first part of the algorithm runs in $d$ steps. In step $i$ we update two linear systems $G^{i+1}$ and $H^{i+1}$ from $G^{i}$ and $H^{i}$ respectively, where initially $G^1 = H^1 = G$. 
The systems $G^i$ (on $d-i+1$ variables) and $H^i$ (on $d$ variables) do basically encode the same solution system, we will see later why a distinction is necessary. During the execution of the algorithm, $G^i$ is used to do Fourier-Motzkin elimination method, $H^i$ is used to run the algorithm of Theorem \ref{proc_AS}. We denote by FM$(x_i, G)$ the set of inequalities obtained by eliminating $x_i$ from $G$ by using one step of the Fourier-Motzkin elimination method.

For any two variables $x_j$ and $x_k$ in $G$, with $j < k$ we represent the set of inequalities containing $x_j$ and $x_k$  (with nonzero coefficients) in the $(x_j, x_k)$ plane as two envelopes, the upper envelope and the lower envelope. The feasible region of $x_j$ and $x_k$ is contained between the envelopes (in the pink region of Figure \ref{fig_case1}) and each envelope can be represented as a piecewise linear function with breakpoints. The projection of the breakpoints onto the $j$-axis is denoted by $B^j_{k}(G)$. If the envelope is unbounded in the $x_j$-direction we add points at infinity to $B^j_{k}(G)$. The range of $x_j$ is hence contained in the interval given by the leftmost and rightmost point of $B^j_{k}(G)$.

\noindent Below follows the pseudo code and the explanation of the algorithm. 

\medskip
\noindent
\begin{tabbing}
123 \= 123 \= 123 \= 123 \= 123 \= 123 \= 123 \= 123 \= \kill
\> Algorithm Hochbaum-Naor ($G$);\\
\> {\bf begin }\\
\> \> $G^1 = H^1 =G$;\\
\> \> {\bf for} $i = 1, \dots, d-1$ {\bf do}\\
\> \> \> Generate $B^i = (b_1^i, \dots, b_m^i)$, the sorted sequence of the points in   \\
\> \> \> $\bigcup_{x_j \in N(x_i, G_i)} B^i_{j}(G^i)$; \\
\> \> \> Use Theorem \ref{proc_AS} to do binary search on $B^i$; \\
\> \> \> {\bf if}  $\exists \ell \in \{1,\dots, m\}$ such that $x_i^{\text{min}}(H^i) \leq b_\ell \leq x_i^{\text{max}}(H^i)$ {\bf then} \\
\> \> \> \> $G^{i+1} := G^i |_{x_i = b_\ell}$;\\
\> \> \> \> $H^{i+1} := H^i \cup \{x_i \leq b_\ell\} \cup \{x_i \geq b_\ell\}$;\\
\> \> \> {\bf else if}  $\exists \ell  (1 \leq \ell < m)$ s.t.\ $b_\ell< x_i^{\text{min}}(H^i)$ and $x_i^{\text{max}}(H^i) < b_{\ell +1}$ {\bf then}\\
\> \> \> \> $G^i := \text{rel}(G^i) \cup \{x_i \geq b_\ell\} \cup \{x_i \leq b_{\ell+1}\}$; \\
\> \> \> \> $G^{i+1} := \text{FM}(x_i,G^i)$; \\
\> \> \> \> $H^{i+1} := H^i$; \\
\> \> \>  {\bf else}  /* system infeasible */ {\bf then}\\
\> \> \> \>  output \emph{system infeasible}; \\
\> \> \> {\bf endif};\\
\> \> {\bf endfor}; \\
\> {\bf end}.\\       
\end{tabbing}

Here rel$(G^i)$ denotes the so called \emph{relevant} inequalities in $G_i$ which we obtain by removing some redundant inequalities. The exact definition follows in the description of the algorithm below.

In step $i$ the algorithm has computed $G^i$ and $H^i$, where originally $G^1 = H^1 = G$. For every pair $(x_i,x_j)$, such that $x_j$ is a neighbor of $x_i$ in $G^i$, i.e., $x_j \in N(x_i, G_i)$, it computes the projections of the breakpoints $B^i_j(G^i)$. The union of those points are sorted and denoted by $B^i$. The idea is now to run a binary search on $B^i$ using Theorem \ref{proc_AS}, in the hope of finding a point in the range of $x_i$. 

If the algorithm finds a breakpoint $b_\ell \in B^i$ such that $x_i^{\text{min}}(H^i) \leq b_\ell \leq x_i^{\text{max}}(H^i)$, 
then we set $x_i := b_\ell$ (see Figure \ref{fig_case1}). We set $H^{i+1} = H^i \cup \{x_i \leq b_\ell\} \cup \{x_i \geq b_\ell\}$ and $G^{i+1} := G^i |_{x_i = b_\ell}$. 

\begin{figure*}
\begin{center}
\includegraphics{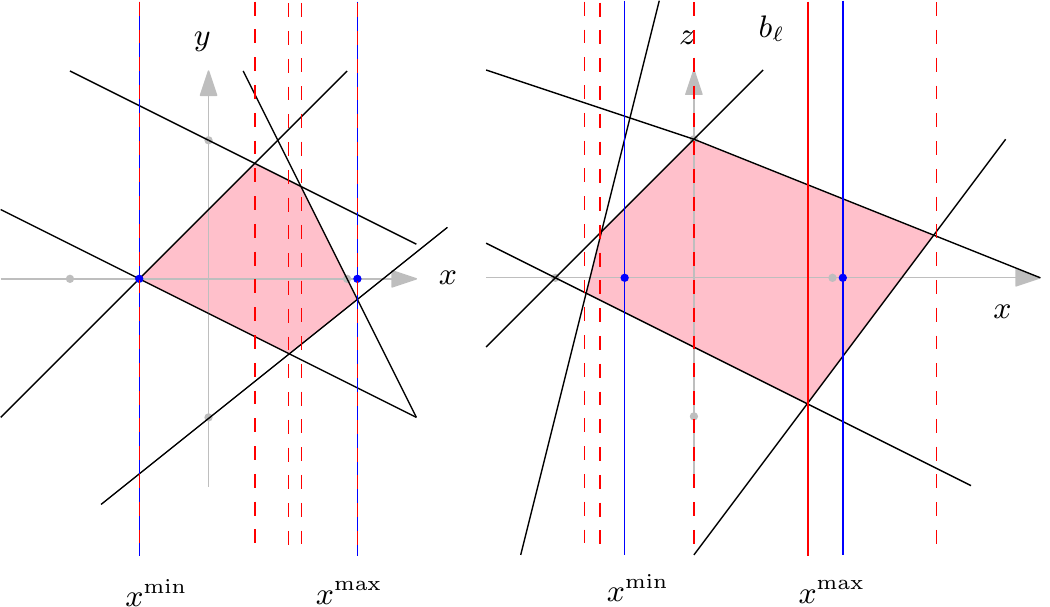}
\end{center}
\caption{First case Hochbaum-Naor, example with 3 variables} \label{fig_case1}
\end{figure*}

If there is no such $b_\ell$, the algorithm finds  $b_\ell$, $b_{\ell + 1}$ such that \ $b_\ell< x_i^{\text{min}}(H^i)$ and $x_i^{\text{max}}(H^i) < b_{\ell +1}$. In that case for any neighbor $x_j$ of  $x_i$ the number of inequalities containing both is reduced to at most two (bold lines of \ref{fig_case2}), the ones that define the upper and lower envelope respectively on the interval $[b_{\ell}, b_{\ell+1}]$ (blue part of Figure \ref{fig_case2}).  This can be done since $[ x_i^{\text{min}},  x_i^{\text{max}}] \subseteq [b_{\ell}, b_{\ell+1}]$ and therefore all other inequalities are redundant and can be removed. We denote the set of inequalities obtained after the removal of the redundant ones by rel$(G^i)$.
The normal Fourier-Motzkin elimination is applied on \text{rel}$(G^i) \cup \{x_i \geq b_\ell\} \cup \{x_i \leq b_{\ell+1}\}$ to obtain $G^{i+1}$. By the above discussion  \text{rel}$(G^i) \cup \{x_i \geq b_\ell\} \cup \{x_i \leq b_{\ell+1}\}$  has the same solution space as $G^i$. As the number of inequalities between $x^i$ and any neighbor $x^j$ is reduced to at most two, the algorithm adds at most four inequalities between any pair of neighbors of $x^i$. This prevents the usual quadratic blowup of the Fourier-Motzkin Method. The system $H^{i+1}$ does not need to be updated, i.e., $H^{i+1} = H^i$. 

\begin{figure*}
\begin{center}
\includegraphics{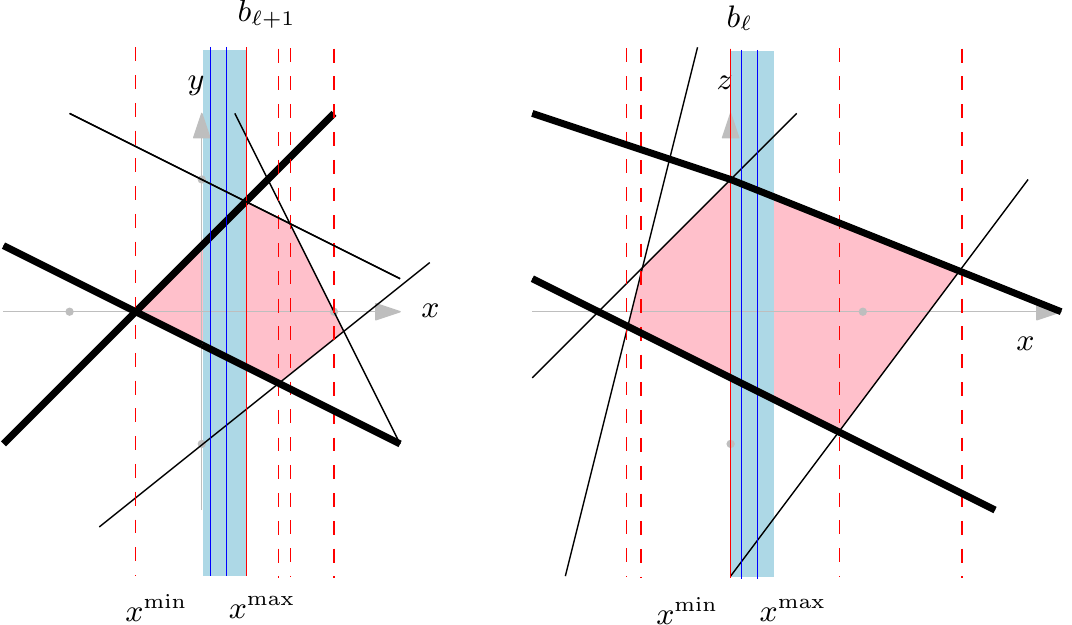}
\end{center}
\caption{Second case Hochbaum-Naor, example with 3 variables} \label{fig_case2}
\end{figure*}

We observe that in every step only a constant number of constraints are added to $H^i$, which guarantees the running time of the binary search to be $O(nd\log n)$. The size of $G^i$ can be of order $\Theta(n + d^2)$ (as we may add up to $4d$ constraints in each step), hence running Theorem \ref{proc_AS} on $G^i$ would not guarantee the running time in case where $n = o(d^2)$.

\paragraph{Second Part} The second part of the algorithm is now the normal backtracking of the Fourier-Motzkin Method. Assume that the variables that were eliminated (in the elseif-step) are $x_{i1}, \dots, x_{ik}$, where $k \leq d$ and $i1 < i2 < \dots < ik$. In the end of part one we are left with the system $G^d$ on variable $x_d$.  By the properties of the Fourier-Motzkin elimination $G^d$ is feasible and the range of $x_d$ is the same as its range in $H^d$. Now choose a feasible value of $x_{d}$ and continue inductively by backtracking through $G^{ik}, \dots , G^{i1}$. 
The geometric interpretation is similar to the first part: for each variable we pick a value in its current range. 

\subsection{Discussion of the Algorithm}
\label{sec_feasibleproof}
We briefly discuss the main points of the proof of Theorem \ref{thm_HN} (for more detail see \cite{HN}).
\begin{proof}[Proof sketch of Theorem \ref{thm_HN} for the feasible case]
Building and updating all enve-lopes takes $O(dn \log n)$ time per step, hence $O(d^2n \log n)$ in total.   Since for all $i$, the size of $H^i$ is $O(n)$ and $|B^i| \leq n + 4d$ ,
in each step the binary search runs Theorem \ref{proc_AS} $O(\log n)$ times, where each evaluation takes time $O(dn)$. It follows that the first part of the algorithm takes time $O( d^2 n \log n)$. For the second part consider the step where we find a solution for a variable $x_j$ in the backtracking step in $G_i$. Then $x_j$ shares at most two inequalities with each of its neighbors, therefore the whole second part only takes time $O(d^2)$. 

During the whole algorithm, the variable $x_i$ is set to some value $\lambda$ if and only if $\lambda$ is in the current bounds $[x_i^{\min}, x_i^{\max}]$. Therefore in the feasible case, it correctly outputs a feasible point of $Ax \leq b$. 
\end{proof}

\subsection{Discussion of the Algorithm in the Infeasible Case}
\label{sec_infeasible}
In the previous Section we showed that if $Ax \leq b$ is feasible, then the Hochbaum-Naor Method always correctly outputs a feasible point. We now show that in the infeasible case, infeasibility is always detected, which completes the proof of Theorem \ref{thm_HN}.

\begin{proof}[Proof sketch of Theorem \ref{thm_HN} for the infeasible case]
Assume that $Ax \leq b$ is infeasible. We now run the first part of the algorithm as in the feasible case. If during the execution at some point during binary search, we detect a contradiction in form of $x_i^{\max} < x_i^{\min}$ this is a certificate for infeasibility and we are done. It is however possible that in every call of the algorithm of Theorem \ref{proc_AS} we get some (wrong) output 	$\lambda < x_i^\text{min}$,
	$\lambda \in [x_i^\text{min}, x_i^\text{max}]$,
	$\lambda > x_i^\text{max}$. 
In that case $G^d$ is an infeasible system on the variable $d$. This follows since the Fourier-Motzkin elimination is feasibility preserving and setting some variables to fixed values in an infeasible system, keeps the system infeasible. Detecting infeasibility in a system with one variable can be trivially done in linear time in the number of constraints. It follows that infeasibility is always detected, which concludes the proof of Theorem \ref{thm_HN}.
\end{proof}

\section{Modification of the Hochbaum-Naor Method}
\label{sec_HNmod}
We now show how to modify the Hochbaum-Naor Method from Section \ref{sec_HN}, such that it decides full-dimensio-nality of the problem and in the full-dimensional case outputs an interior point (see Theorem \ref{thm_modHN}).
For this we need some preparatory lemmas.
\begin{lemma}
\label{lemma_fullHN}
Let $Ax \leq b$ feasible and let $\lambda \in (x_1^{\min}, x_1^{\max})$. Then $y:= (\lambda, y_2, \dots, y_d)$ is an interior point solution of $Ax \leq b$ if and only if $y':=(y_2, \dots, y_d)$ is an interior point solution of $A'x \leq b'$, where $A'x \leq b' $ is the system obtained by $(Ax \leq b) |_{x_1 = \lambda}$. 
\end{lemma}

\begin{proof}
Let $y$ be an interior point solution of $Ax \leq b$. Then by definition $Ay < b$ and obviously $A'y' < b'$. On the other hand let $y'$ be an interior point solution of  $A'x < b'$, i.e., $A'y' <b'$. Then $y$ satisfies any inequalities containing some $x_i \neq x_1$ strictly. The only inequalities that might be satisfied with equality are the ones containing only $x_1$, but this is a contradiction to $\lambda \in (x_1^{\min}, x_1^{\max})$.
\end{proof}

\begin{lemma}
\label{lemma_fullFM}
In the Fourier-Motzkin Method $G_{i+1}$ has an interior point solution, if $G_i$ has one. 
Moreover if an interior point solution of $G_1 = G$ exists, it can be obtained in the running time of the Fourier-Motzkin algorithm.
\end{lemma}
\begin{proof}
The first part follows by Lemma \ref{lemma_fullHN}. 
For the second part is is enough to consider a slight variant of the Fourier-Motzkin elimination. Instead of running the algorithm on a system $Ax \leq b$, we run it on $Ax < b$. In each step the inequalities obtained are of form $l < h$ instead of $l \leq h$. By induction, using Lemma \ref{lemma_fullHN} one can see that finding a solution using this variant, is equivalent to finding an interior point of $Ax \leq b$.   
\end{proof}

\begin{proof}[Proof of Theorem \ref{thm_modHN}]
Assume that $Ax \leq b$ is feasible. We run the Hochbaum-Naor Algorithm almost in the same way as described in Section \ref{sec_HN}. The only difference is in the if-loop of  the algorithm. The original algorithm distinguishes between the cases $\lambda < x_i^{\text{min}}$, $\lambda \in [x_i^{\text{min}}, x_i^{\text{max}}]$ and $\lambda > x_i^{\text{max}}$ of Theorem \ref{proc_AS}. Our algorithm however distinguishes between $\lambda \leq x_i^{\text{min}}$, $\lambda \in (x_i^{\text{min}}, x_i^{\text{max}})$ and $\lambda \geq x_i^{\text{max}}$.

In the first case we only set $x_i =b_\ell$ if there exists a breakpoint $b_\ell \in B$, such that $x_i^{\text{min}}(H^i) < b_{\ell} < x_i^{\text{max}}(H^i)$. We only fix $x_i$ to some value $b_\ell$, if $b_\ell$ is in the \emph{open} range $(x_i^{\text{min}}(H^i), x_i^{\text{max}}(H^i))$, (in the original Theorem we were considered the closed range). The second case accordingly changes to finding an 
interval $[b_\ell, b_{\ell + 1}]$ ($1 \leq \ell <k$) such that $b_\ell \leq x_i^{\text{min}}(H^i)$ and $x_i^{\text{max}} (H^i) \leq b_{\ell + 1}$, (originally $b_\ell < x_i^{\text{min}}(H^i)$ and $x_i^{\text{max}} (H^i) < b_{\ell + 1}$).  We see that in this case, the number of inequalities on each edge adjacent to $x_i$ is still reduced to at most two.

As Theorem \ref{proc_AS} distinguishes the cases $\lambda < x_i^\text{min}$,
	$\lambda = x_i^\text{min}$,
	$\lambda \in (x_i^\text{min}, x_i^\text{max})$,
	$\lambda = x_i^\text{max}$
	$\lambda > x_i^\text{max}$ and certificate for infeasibility in time $O(nd)$, the running time remains the same.

It remains to show that the modified algorithm detects full-dimensionality and in the full-dimensional case an interior point.
 
 The discussion of the case where $Ax \leq b$ is infeasible is equivalent as in the proof of Theorem \ref{thm_HN}.
Hence assume that the system is feasible.

Let $Ax \leq b$ be full-dimensional. 
By Lemma \ref{lemma_fullHN}, after the first part of the algorithm $H^d$ (and $G^d$) are associated with a system of inequalities whose interior point solutions can be extended to an interior point solution of $Ax \leq b$. 
The interior point solution of $Ax \leq b$ can now be found in the backtracking step using Lemma \ref{lemma_fullFM}. Assume such a point can not be found. Then by Lemma \ref{lemma_fullFM} there is no interior point of $Ax \leq b$, which is a contradiction to full-dimensionality.

Let $Ax \leq b$ be feasible but non-full-dimensional. Then at some point of the backtracking the algorithm finds $x_i^{\min} = x_i^{\max}$ for the current bounds. Suppose this does not happen, then by Lemma \ref{lemma_fullHN} the algorithm finds an interior point, which is a contradiction to non-full-dimensionality.

\end{proof}

\section{The Non-Full-Dimensional Case}
\label{sec_nonfull}

In the non full-dimensional case redundancies are dependent on each other, meaning that a redundant constraint can become nonredundant after the removal of another redundant constraint. The problem is therefore to find a maximal set of nonredundant constraints. 

Clarkson's Algorithm can be extended for redundancy removal in the non-full-dimensional case as follows: In a preprocessing step one can find the dimension $k$ of the system $Ax \leq b$, by solving at most $d$ linear programs \cite{f-lnpc-12}. Of all the inequalities that are forced to equality, we can find a set of $(d-k)$ equalities that defines the $k$-dimensional space where  $P = \{x \mid Ax \leq b\}$ lies in. Let us denote the remaining system of inequalities (the ones not forced to equality) by $A_2x \leq  b_2$. One can now rotate the the system such that $P$ lies in $R^k$. Clarkson's algorithm can now be applied in $R^k$, where the constraints are the intersections of the rotated system of $A_2x \leq  b_2$ intersected with $R^k$. After the preprocessing the running time is hence $O(n \cdot LP(s,k))$. 

In the case of $LI(2)$ we observe that such a rotation may destroy the structure of two variables per inequality. It results that we are still able to match Clarkson's running time, using substitution of variables. 

\begin{theorem}
\label{thm_nonred}
Let $Ax \leq b$ a $k$-dimensional system in $LI(2)$, for $0 \leq k \leq d$. Then given a relative interior point solution of $Ax \leq b$, all redundancies can be detected in time $O(nk^2s\log s + d^2n)$. 
\end{theorem}

The $d^2n$ term comes from Gaussian elimination, which is dominated by the preprocessing time needed to find the relative interior point (see Proposition \ref{prop_relint}). Note that the typically larger term $nk^2s\log s$ is dependent on the dimension of the polytope $k$ and not $d$.

We need the following observation for the proof of Theorem \ref{thm_nonred}.

\begin{observation}
 \label{obs_linind}
Let $Ax \leq b$ in $LI(2)$ and $\alpha x_i + \beta x_j \leq \gamma$, $\beta \neq 0$, an inequality of the system that is forced to equality, i.e., $\alpha x_i^*+ \beta x_j^* = \gamma$, for all solutions $x^*$. Let $\overline{A}x \leq \overline{b}$ be the system obtained by substituting $x _j = \frac{\gamma}{\beta} - \frac{\alpha}{\beta}x_i$. Then the following holds.
\begin{itemize}
\item
$\overline{A}x \leq \overline{b}$ is still in LI(2).
\item
A constraint is redundant in $Ax \leq b$ if and only if it is redundant in the system $\overline{A}x \leq \overline{b}$.
\end{itemize}

\end{observation}

\begin{proof}
Given a relative interior point $x^*$, one can find $A_1x \leq b_1$, the subsystem of $Ax \leq b$ that is forced to be equality, in time $O(nd)$.  The remaining system is denoted by $A_2x \leq b_2$.
Finding a minimal subsystem $A_1^*x = b^*_1$ of  $A_1x = b_1$ with $(d-k)$ linearly independent equalities   that defines the $k$-dimensional space containing  $P = \{x \mid Ax \leq b\}$, takes $O(d^2n)$ time using the Gaussian elimination. 
Using these equalities we can substitute $d-k$ variables of $A_2x \leq b_2$ in the same fashion as explained in Observation \ref{obs_linind}. Hence we get a $k$-dimensional representation of $Ax \leq b$ which is in $LI(2)$.

We can now run the algorithm given by Theorem \ref{thm_main} on $A_2'x \leq b_2'$, the system obtained from $A_2x \leq b_2$ after substitution. These detected nonredundant constraints together with $A_1^*x = b_1^*$ give us a minimal set of nonredundant inequalities.
\end{proof}

The following proposition shows that finding a relative interior point can also be done in strongly  polynomial time. 
\begin{proposition}
\label{prop_relint}
Given $Ax \leq b$, one can find a relative interior point of $Ax \leq b$ or a certificate for infeasibility in time $O(d^2 n \log n)$.
\end{proposition}

\begin{proof}
Using a similar argument as in Lemma \ref{lemma_fullHN}, one can show the following.
\begin{itemize}
	\item If $x_1^{\min} < \lambda < x_1^{\max}$, then $y: = (\lambda, y_2, \dots, y_d)$ is a relative interior point of $Ax \leq b$ if and only if $y' := (y_2, \dots, y_d)$ is a relative interior point of $(Ax \leq b) |_{x_1= \lambda}$.
	 \item If $x_1^{\min}= x_1^{\max}=\lambda$, then $y: = (\lambda, y_2, \dots, y_d)$ is a relative interior point of $Ax \leq b$ if and only if $y' := (y_2, \dots, y_d)$ is a relative interior point of $(Ax \leq b) |_{x_1= \lambda}$.
\end{itemize}
The algorithm for finding a relative interior point is very similar to the modified Hochbaum-Naor Method. The first part runs equivalently. In the second part of the backtracking if at some point $x_i^{\min} = x_i^{\max}$, we set $x_i: = x_i^{\min} = x_i^{\max}$. Infeasibility is detected by the same argument as in Theorem \ref{thm_HN}. Correctness follows from the above argument, the running time is the same as in Theorem \ref{thm_modHN}. 
\end{proof}

\begin{corollary}
\label{cor_dim}
The dimension of $Ax \leq b$ or a certificate for infeasibility can be found in time $O(d^2 n \log n)$, i.e., the same running time as finding a feasible point solution of a certificate for infeasibility.
\end{corollary}

\begin{proof}
Consider the algorithm of the proof of Proposition \ref{prop_relint}. Since we know that infeasibility will be detected, assume that $Ax \leq b$ is feasible.
Let us denote the current polytope by $P$, where initally $P= \{x \mid Ax \leq b\}$ the polytope defined by $Ax \leq b$. Every time $x_i$ is set to a value in the open range $(x_i^{\min}, x_i^{\max})$, the dimension of the current polytope $P$ decreases by 1, as we intersect it with a hyperplane not containing $P$. If $x_i^{\min} = x_i^{\max}$ and we set $x_i := x_i^{\min} = x_i^{\max}$, then the dimension of the current polygon stays the same, as we intersect it with a hyperplane containing $P$. Since after setting all to some value we end up with a point (polygon of dimension 0), the dimension of $Ax \leq b$ is exactly the number of times we set $x_i$ to a value in the open range $(x_i^{\min}, x_i^{\max})$.
\end{proof}

\section{Acknowledgments}
The authors would like to thank Seffi Naor for insights into the problem.
Moreover we would like to thank Jerri Nummenpalo and Luis Barba for interesting discussions. We are especially grateful to Jerri Nummenpalo for making us aware of some important literature and helping us during the writeup of this paper. 
\bibliographystyle{plain}
\bibliography{redundancyFVC}

\end{document}